\documentclass[twoside,10pt]{amsart}
\usepackage{graphicx,amsmath,amsfonts,amsthm,amssymb, enumerate,float, fancyhdr,verbatim}

\newtheorem{thm}{Theorem}[section]
\newtheorem{cor}[thm]{Corollary}
\newtheorem{lem}[thm]{Lemma}
\newtheorem{prop}[thm]{Proposition}
	
\theoremstyle{definition}
\newtheorem{defn}[thm]{Definition}	
\theoremstyle{remark}

\def\beq{\begin{eqnarray}}
\def\eeq{\end{eqnarray}}
\def\bsp{\begin{split}}	
\def\esp{\end{split}}

\newcommand{\be}{\begin{equation}}
\newcommand{\ee}{\end{equation}}

\newcommand{\z}{\zeta}
\newcommand{\bz}{\bar{\zeta}}

\newcommand{\ba}{\bar{a}}

\begin{document}

\title{\Large\textbf{Vacuum Plane Waves; Cartan Invariants and physical interpretation}}
\author{{\large\textbf{A. Coley$^{1}$ }, D. McNutt$^{1}$, and R. Milson$^{1}$ }
%EndAName
%\address{
 \vspace{0.3cm} \\
$^{1}$Department of Mathematics and Statistics,\\
Dalhousie University,
Halifax, Nova Scotia,\\
Canada B3H 3J5
\vspace{0.2cm}\\
\texttt{aac@mathstat.dal.ca,ddmcnutt@dal.ca,rmilson@dal.ca } }
\date{\today}
\maketitle
\pagestyle{fancy}
\fancyhead{} % clear all header fields
\fancyhead[EC]{D. McNutt, R. MIlson and A. Coley}
\fancyhead[EL,OR]{\thepage}
\fancyhead[OC]{Vacuum Plane Waves}
\fancyfoot{} % clear all footer fields

\begin{abstract} 
As an application of the Cartan invariants obtained using the Karlhede algorithm, we study a simple subclass of the PP-wave spacetimes, the gravitational plane waves. We provide an invariant classification of these spacetimes and then study a few notable subcases: the linearly polarized plane waves, the weak-field circularly polarized waves, and another class of plane waves found by imposing conditions on the set of invariants. As we study these spacetimes we relate the invariant structure (i.e., Cartan scalars) to the physical description of these spacetimes using the geodesic deviation equations relative to timelike geodesic observers.  
\end{abstract} 

\maketitle

\begin{section}{The Plane wave spacetimes and Cartan Invariants} 

The plane waves were introduced by Rosen \cite{Rosen} in 1937 to
describe wave-like solutions to the Einstein equations.  However, due to
the choice of coordinates, Rosen concluded that these metrics were
unphysical due to singularities in the metric components.  Upon further
analysis these singularities were shown to be coordinate dependent and
easily eliminated by a change in coordinates, \cite{Bondi, BPR}.  In
1961 the plane waves were shown to belong to the class of PP-wave
spacetimes\footnote{These are Petrov type N spacetimes admitting a
covariantly constant null vector, $\ell$.}  \cite{Kundt61, Witten}
describing pure radiation far from an isolated source; these were
originally studied by Brinkmann in 1925 as a subcase of all
N-dimensional Einstein spacetimes which are related by conformal
transformations with vanishing Ricci scalar\cite{Brinkmann}.  Despite
the existence of closed null geodesic curves \cite{Penrose}, these
spacetimes have been studied in classical general relativity as well as in its
generalizations \cite{Bicak, Ozvath, coley}.

In a PP-wave spacetime, all polynomial scalar invariants vanish
\cite{4DVSI}; therefore, to classify such spacetimes we need to apply
the Karlhede equivalence method \cite{4DCSI, 4DKundt}.  To start, one
chooses a canonical null tetrad where the Weyl tensor component has been
normalized (i.e., $\Psi_4 = 1$) by making a specific Lorentz spin and boost.
The first order Cartan invariants in the Karlhede algorithm arise as
components of the first and second order covariant derivatives of the
Weyl tensor ${\bf \Psi}$.  As $\Psi_4$ is constant, these additional
invariants take the form of the spin-coefficient $\alpha$, $\gamma$, and
their conjugates $\bar{\alpha}, \bar{\gamma}$, which are introduced at
first order as components of the covariant derivative of the Weyl tensor
\cite{milson}.

To study the plane waves, we assume $\alpha =0$.  We add a subscript $c$
to indicate the fact that the non-zero spin coefficients relative to
this class of canonical coframes (in which $\Psi_4 =1$) are Cartan
invariants as well.  Following the analysis in \cite{milson}, at second
order the plane waves offer only one new Cartan invariant, $\Delta
\gamma_c$, since all of the remaining second order invariants vanish; i.e.,
$\mu_c = \nu_c = \bar{\delta} \alpha_c = \delta \gamma_c = D \gamma_c =
0$, with $\Delta \gamma_c$ potentially non-zero.  Unlike those PP-wave
spacetimes with $\alpha \neq 0$, it is not possible to produce an invariant coframe
in the plane wave spacetimes as the Riemann tensor and its covariant
derivatives to all orders are invariant under null rotations about
$\ell$.  Thus the coframe we use to produce the Cartan invariants is not
unique, since any null rotation about $\ell$ produces a new coframe with
$\gamma_c$ and $\Delta \gamma_c$ unchanged.  Despite the lack of
invariant coframe, $\Delta \gamma_c$ is we ll defined due to this
invariance under null rotations.

From these facts we have a helpful proposition to classify the plane wave spacetimes:
\begin{prop} \label{prop:gammaInv}
A plane wave spacetime may be  locally described in an invariant manner using the triplet of Cartan invariants $\{\gamma_c, \bar{\gamma}_c, \Delta\gamma_c \}$, where the remaining two invariants are expressed in terms of $\gamma_c$ and hence do not change in any coordinate system. \end{prop}

\noindent In the case where $\gamma_c$ is non-zero and constant, then it may be shown that these are the $G_6$ spacetimes given in Table \ref{table:ppcase2App}. If $\gamma_c$ is non-constant, the vanishing of $\delta \gamma_c, \bar{\delta} \gamma_c$ and $D \gamma_c$ implies  $\gamma_c$ is at best a function of one variable; instead of the complex-valued function $\gamma_c$ we may take the real or imaginary component of $\gamma_c$ as the functionally independent invariant and the classifying invariants $\bar{\gamma}_c$ and $\Delta \gamma_c$ are replaced with real valued functions.  
%From this argument we have a helpful corollary:
%\begin{cor} \label{cor:gamma12cord}
%If $\gamma_1 = Re(\gamma_c)$ (or $\gamma_2 = Im(\gamma_c)$) is non-constant and non-zero, one may write the null coordinate, $u$, in terms of $\gamma_1$ ( or $\gamma_2$) , i.e. $u = U(\gamma_{1})$ (or $\tilde{U}(\gamma_2)$). 
%\end{cor}

To see this, consider the Brinkmann coordinates for the vacuum plane wave metric \cite{Podolsky},
 \beq -2du dv - 2H(u)du^2 + 2d\z d \bz,~~H(u,\z,\bz) = Re(f(\z,u)). \label{KundtPlaneWave} \eeq
\noindent In this coordinate system, relative to the frame in which $\Psi_4 =1$, the vanishing spin coefficient $\alpha_c$ implies
\beq \alpha_c = e^{a-\ba} \ba_{,\bz} = 0, \label{AlphaBrinkmann} \eeq  
\noindent where $a = \frac14 ln H_{,\z\z}$; clearly $\ba_{,\bz} = 0$, so that $\bar{f}_{,\bz \bz \bz} = 0$, giving a solution of the form \beq f(\z,u) = A(u) \z^2. \label{PWBrinkmanf} \eeq
\noindent Expressing $\gamma_c$ in these coordinates: 
\beq \gamma_c = \frac{1}{4 \sqrt{{A\bar{A}}}} ln(\bar{A})_{,u},  \label{PWBrinkmanGamma} \eeq
\noindent where a particular coordinate system has been used. However, 
regardless of which coordinate system is used, we may calculate $\gamma_c$ 
in the canonical coframe, as only Lorentz transformations are used.  
Supposing that $\gamma_c = \gamma_1 + i \gamma_2$ with $A = r(u)e^{i\theta(u)}$, we may solve for $A$ in Brinkmann coordinates.
\begin{lem} \label{lem:G5App}
For any PP-wave spacetime expressed in terms of a canonical coframe with $\alpha_c = 0$ and $\gamma_c = \gamma_1 + i \gamma_2~~\neq 0$, we may express the canonical form for $f(\z,u)$ as 
\beq &A= re^{i\theta};~~~ r(u) = [C_0 - \int 4\gamma_1 du]^{-1},~~\theta(u) = - \int 4 r \gamma_2 du+C_1,~~~C_0,C_1 \in \mathbb{R}.& \label{1G5App} \eeq
\end{lem}

Here, $\gamma_c$ gives rise to the only functionally independent invariant and the essential classifying functions are $\bar{\gamma}_c(u)$ and $\Delta \gamma_c (u)$ expressed in terms of $\gamma_1$. If $\gamma_c$ is constant, there are two possibilities for $A(u)$ depending on where $\gamma_c$ lies in the complex plane, which are given in table \ref{table:ppcase2App}.

\begin{table}[H] 
\begin{center} % used for centering table
\begin{tabular}{c|c c c }
\hline 
 & $f(\z,u)$ & $1$ & $2$   \\ [0.5ex]
\hline \\
$G_5$ & $\frac{A(u)}{2} \z^2$, \eqref{PWBrinkmanf}  & $\gamma_c$; \eqref{1G5App} & $\Delta \gamma_c $ \\ [1ex]
\hline \\
$G_6$-a & $\frac{u^{\frac{iC_1}{C_0} -1}}{16C_0^2}\z^2$ & $\gamma_c= C_0 + i C_1$ &  \\ [1ex]
$G_6$-b & $e^{iC_1u} \z^2$ & $\gamma_c = i\frac{C_1}{4}$ & \\ [1ex]
\hline \\
\end{tabular}
\caption{ Summary of cases with $\alpha_c = 0$. $C_0,C_1 \in \mathbb{R}$, and $A(u)$ is a complex valued function. }
\label{table:ppcase2App}
\end{center}
\end{table}

As an example, we provide an invariant description for the class of spacetimes for which all timelike geodesic observers produce a linear polarization in terms of the equations of geodesic deviation \cite{Podolsky}, which will be discussed in the following section. For now we make the following definition in terms of Cartan invariants:
\begin{defn} \label{dfn:glinPolar}
A vacuum plane wave spacetime is {\it linearly polarized} when the Cartan invariant, $\gamma_c$, is real-valued. 
\end{defn}
\noindent Applying Lemma \eqref{lem:G5App} we find a particular form for the linearly polarized waves: 
\begin{cor} \label{cor:glinPolar}
 Given a vacuum plane wave spacetime, relative to the class of canonical coframes where $\Psi_4 = 1$ suppose $\bar{\gamma}_c = \gamma_c$. The metric expressed in Brinkmann coordinates has $A$ real-valued and \beq &A = [C_0 - \int 4\gamma_c du]^{-1}.& \nonumber \eeq 
\end{cor}
\noindent We will say the plane wave is $+$ linearly polarized. Making a  rotation in the spatial coordinates $\z' = e^{i\pi/4}\z$ (equivalently a spin in the transverse plane), the metric function is multiplied by $i$ so that $\Psi_4 = i$, and we say this is $\times$ linearly polarized. For more general polarization states the triplet of invariants $\{ \gamma_c(u), \bar{\gamma}_c(u), \Delta \gamma_c (u) \}$ describes how the $+$ and $\times$ polarization states mix.
\end{section} 
\begin{section}{The Equations of Geodesic Deviation - Polarization modes for all PP-wave spacetimes } \label{GDsec}

To study the polarization modes of gravitational waves in vacuum
spacetimes with a cosmological constant, a particular null tetrad is
introduced relative to the Brinkmann coordinates.  This null tetrad
arises in the choice of an orthonormal frame in which the equations of
geodesic deviation,
\beq \frac{d^2 Z^{\mu}}{d \tau^2} = \ddot{Z}^{\mu} = -R^{\mu}_{~\alpha \beta \gamma}u^{\alpha}Z^{\beta}u^{\gamma},\label{GDequations} \eeq 
\noindent 
take on a simpler form, where $\dot{{\bf x}} = d {\bf x}/d\tau$, $|{\bf x}|^2 = -1$ is the four velocity of a timelike geodesic curve corresponding to a free test particle, $\tau$ is the proper time along this curve and $Z(\tau)$ is a displacement vector perpendicular to $\dot{{\bf x}}$. 

To construct the desired null tetrad, we first produce an orthogonal frame with $\dot{{\bf x}} = e_1$ and the remaining vectors $\{ e_2,e_3,e_4 \}$ from the local hypersurface orthogonal to $e_1$ (so that $<e_a, e_b> = g_{\alpha}{\beta}e^{\alpha}_a e^{\beta}_b = \eta_{ab}$). The dual basis will be $e^1 = -\dot{{\bf x}}$ and $e^i = e_i,i=2,3,4$. This will hold at a point along the timelike geodesic $x^{\mu}(\tau)$. If we wish to have this hold on the entire curve the coframe must be parallely transported along the curve, yielding further conditions on the components of the metric and the four-velocity $\dot{{\bf x}}(\tau)$.

Choosing Brinkmann coordinates, so that the metric takes the form
\eqref{KundtPlaneWave}, the plane waves are further constrained as the
analytic function must be of the form, $f(\z,u) = A(u) \z^2$.  These
solutions admit an isometry group of dimension five in general and dimension six
if and only if $\gamma$ is constant.

The PP-wave spacetimes belong to the subclass of $KN(\Lambda)[\alpha', \beta']$ \cite{Ozvath} with $\Lambda = 0$ and where the arbitrary functions $\alpha'$ and $\beta'$ may be set to $\alpha' =1$, $\beta' = 0$ via an appropriate coordinate transform preserving the metric form. We project the geodesic deviation equations onto this orthonormal frame in the case $\Lambda = 0$ and $\Psi_4 \neq 0$:
\beq \ddot{Z}^1 = 0,~~ \ddot{Z}^{2} = - A_{+} Z^2 + A_{\times} Z^3,~~\ddot{Z}^{3} = A_{+} Z^3 + A_{\times} Z^2,~~\ddot{Z}^{4} = 0 \label{GDortho} \eeq
\noindent Where the dot above a function denotes differentiation with respect to 
the proper time $\tau$ along the geodesic and \beq  A_{+} \equiv \frac14 (\Psi_4 + \bar{\Psi}_4),~~A_{\times} \equiv \frac{i}{4} (\bar{\Psi}_4 - \Psi_4). \label{ReImP4} \eeq 
\noindent Using the null tetrad,  
\beq m_i = \frac{1}{\sqrt{2}} (e_2 + ie_3),~~n_i = \frac{1}{\sqrt{2}}(e_1-e_4), ~~\ell_i = \frac{1}{\sqrt{2}}(e_1+e_4), \label{InterpTet} \eeq
\noindent and denoting $Z = z^0 \ell_i + z^1 n_i + z^2 m_i + z^3 \bar{m}_i$ with $\bar{z}^2 = z^3$, the equations of geodesic deviation become:  
\beq & \ddot{z}^0 = 0, \ddot{z}^1 = 0,& \nonumber \\
 & \ddot{z}^{2}  = -(A_{+} + i A_{\times}) z^3 = - \frac12 \Psi_4 z^3.& \label{GDnull} \eeq

To determine the form of the null tetrad \eqref{InterpTet} for the PP-wave spacetimes 
we must choose $\ell$ to be the preffered null direction along which the Weyl tensor 
has one non-vanishing component, $\Psi_4$. With this null direction, we can prove the following proposition \cite{Podolsky}
\begin{prop} \label{prop:1}
 Let $\dot{{\bf x}}$ be the four velocity of a timelike geodesic, and $\ell$ some null vector. Then there exists a unit spacelike vector $e_4$ which is the projection of the null direction given by $\ell$ into the hypersurface orthogonal to $\dot{{\bf x}}$. This spatial vector is unique (up to reflections) and is given by $e_4 = -\dot{{\bf x}} + \sqrt{2} \ell$, where $<\ell,\dot{{\bf x}}> = -\frac{1}{\sqrt{2}}$. Another null vector $n$ in the ($\dot{{\bf x}}, e_4$) plane such that $<\ell,n> = -1$ is then given by $n = \sqrt{2}\dot{{\bf x}} - \ell$. The only remaining freedom are rotations in the $(e_2,e_3)$ plane.   
\end{prop}
\noindent In Brinkmann coordinates, $(\z,\bz,u,v)$, we have two more propositions \cite{Podolsky}, the first of which gives the form of the null tetrad:
\begin{prop} \label{prop:2}
 In Brinkmann coordinates, the null tetrad tied to the 4-velocity of the geodesic, $\dot{{\bf x}} = (\dot{\z},\dot{\bz},\dot{u}, \dot{v})$, takes the simple form
\beq & m_i^{\mu} = \left(- \frac{\dot{\z}}{\dot{u}},0,-1,0 \right),~ \bar{m}_i^{\mu} = \left( -\frac{\dot{\z}}{\dot{u}},-1,0,0 \right), & \nonumber \\
 & \ell_i^{\mu} = \left( \frac{1}{\sqrt{2} \dot{u}},0,0,0 \right),~ n_i^{\mu} = \left( \sqrt{2}\dot{v}-\frac{1}{\sqrt{2} \dot{u}},\sqrt{2}\dot{\z}, \sqrt{2}\dot{\bar{z}}, \sqrt{2} \dot{u}  \right). & \label{InterpTet0} \eeq
\noindent where the function $H$ in the metric \eqref{KundtPlaneWave} is hidden due to the identity 
\beq & 2\dot{\z} \dot{\bz} - 2 \dot{u} \dot{v} - H \dot{u}^2 = -1.& \nonumber \eeq
\noindent {\bf Remark}: The null vector $\ell$ is no longer a covariantly constant null vector, as $\nabla_{\frac{\partial}{\partial x^{\mu}}} \ell_i = - \dot{u} \left(\frac{1}{\dot{u}} \right)_{,\mu} \ell_i$; however, it is a recurrent null vector and there is a covariant constant null vector proportional to the original Killing vector $\ell$.    
\end{prop}
\noindent Of course, for an arbitrary unit timeike geodesic, 
$\dot{{\bf x}} = (\dot{\z}, \dot{\bar{\z}}, \dot{u}, \dot{v})$, we may reconstruct the usual metric coframe 
\beq & m_n^{\mu} = (1,0,0,0),~~ \bar{m}_n^{\mu} = (0,1,0,0),~~ \ell_n^{\mu} = (0,0,0,1),~~n_n^{\mu} = (0,0,1,-H) & \label{KundtTet} \eeq
\noindent from the interpretation tetrad $\{ m_i, \bar{m}_i, \ell_i, n_i \}$ in \eqref{InterpTet0} by applying the following Lorentz transformation
\beq & \ell_n = A\ell_i,~~n_n = A^{-1}(\ell_i + Be^{iV} \bar{m}_i + \bar{B} e^{-i\mathfrak{V}} m_i + B \bar{B} n_i),~m_n = e^{-i\mathfrak{V}}m_i +B \ell_i & \nonumber \\
& A = \sqrt{2} \dot{u},~~B =-\sqrt{2} \dot{\z},~~\mathfrak{V} = \pi & \label{ntot} \eeq
To relate this to a physical description one must have a tetrad that
will be defined on all points along the timelike geodesic curve.  In the
more general $KN(\Lambda)[\alpha', \beta']$ class this requirement
imposes further differential constraints on the metric functions.
Fortunately, in the case of the PP-waves these constraints are trivial
\begin{prop} \label{prop:3} 
 For any timelike geodesic $x^{\mu}(\tau) = (\z,\bz,u,v)$ in a PP-wave spacetime, the null tetrad given by \eqref{InterpTet0}  is parallely transported along this geodesic.
\end{prop}
\begin{proof}
From Proposition 3 in \cite{Podolsky}, the tetrad arising from setting $\Lambda = 0$, $\alpha = 1$ and $\beta = 0$ \footnote{These are defined in \cite{Podolsky} and are not to be confused with the spin-coefficients} via a coordinate transform gives the following conditions for the null tetrad $A4$ in \cite{Podolsky} to be parallel transported along the timelike geodesic: 
 \beq \left( \frac{q}{p} \right)_{,\z} = \left( \frac{q}{p} \right)_{,\bz} = 0,~~ \dot{\mathfrak{V}}(\tau) = i\left( \frac{p_{,\z}}{p} \dot{\z} - \frac{p_{,\bz}}{p}\dot{\bz} \right). \label{p3} \eeq
\noindent In these coordinates $p=q=1$ and so the above vanishes, implying $\mathfrak{V}$ must be a constant as its dot derivative is zero.
\end{proof}

The interpretation tetrad \eqref{InterpTet0} (up to constant spins and
boosts) is the only tetrad which is parallel transported along the
chosen timelike geodesic and provides the simplest form from which one
can determine the polarization of a wave along this timelike geodesic.
However, the geodesic deviation equations are frame dependent.  As a
simple example of this we may show that the magnitude of the wave is
dependent on the timelike observer.  By applying a boost in an arbitrary
direction with constant velocity $(v_1,v_2,v_3)$ we obtain:
\beq \Psi_4' = \frac{(1-v_3)^2}{1-v_1^2-v_2^2-v_3^2} \Psi_4. \label{P4boost} \eeq
\noindent 
Thus, the magnitude of the plane wave is dependent on the timelike observer as well; 
an observer traveling with a higher velocity relative to the original timelike geodesic 
will measure a smaller value for the magnitude of the wave. In fact, setting $v_1=v_2=0$ 
and taking the limit $v_3 \to 1$ implies $\Psi_4' \to 0$. 
%A more elaborate example arises when we wish to express the geodesic deviation equations in terms of the canonical tetrad. To construct the canonical null tetrad from \eqref{InterpTet0} one must apply a spin and boost to normalize $\tilde{\Psi}_4 =1$: 
%\beq & \tilde{\ell} = e^{X} \ell_i,~~\tilde{n} = e^{-X} n_i,~~\tilde{m} = e^{iP} m_i& \nonumber \\
%& e^{2X} = |\Psi^i_4|,~~e^{2IP} = \Psi^i_4/\bar{\Psi}^i_4. & \label{bsInterp} \eeq 
%\noindent As long as we are concerned with points on the timelike geodesic the Weyl tensor component $\Psi_4$ will be a function of $\tau$, because one may substitute $x^{\mu} = (\z(\tau), \bz(\tau), u(\tau), v(\tau))$ to make everything dependent on proper time. In this sense, the two frames may be related to each other along an arbitrary timelike geodesic, 
%\beq \Psi^i_4 = A^2 \Psi^n_4 =  \dot{u}^2 f_{,\z\z}, \label{P4Interp} \eeq
%\noindent where $\Psi^i_4$ and  $\Psi^n_4$  are the Weyl tensor components relative to the interpretation tetrad \eqref{InterpTet0} and natural tetrad \eqref{KundtTet} respectively.

If one is interested in the classification of the plane waves the
canonical coframe and Cartan invariants provide a general classification
that complements the study of the geodesic deviation equations.  As an
illustration of this, we will show that those spacetimes for which the equations of geodesic deviation are linearly
polarized for all
timelike geodesics may be defined in an invariant fashion, and that the + and
'$\times$' linear polarization modes arise as a choice of coordinates.
Taking the non-zero component of the Weyl tensor relative to the
interpretation frame \eqref{InterpTet0} \beq \Psi^i_4 = A^2f_{,\z\z} =
\dot{u}^2 f_{,\z\z}, \label{P4Interp} \eeq
\noindent 
we claim that the function $e^{i2P} = \Psi_4/ \bar{\Psi}_4$ is an invariant 
that is independent of the choice of timelike geodesic (since it lacks 
$\dot{u}$ and all of the other components of the timelike geodesic 
4-velocity\footnote{Glossing over the fact that the coordinates 
$(\z,\bz,u,v)$ may be written as functions of $\tau$ for 
some timelike geodesic with proper time, $\tau$.}).  
Rewriting $e^{2iP}$ and assuming $A_{+} \neq 0$, it is easily shown 
that there is only one real 
function involved:
\beq \Psi_4/\bar{\Psi}_4 = \frac{1 + i\frac{A_{\times}}{A_{+}}}{1-i\frac{A_{\times}}{A_{+}}}. \nonumber \eeq
\noindent Note that if $A_{+} = 0$ the phase is already determined; i.e., $P = \pi/2$ mod $2\pi$. Now if we apply lemma \eqref{lem:G5App} and equation \eqref{1G5App} from the previous section, we note that in the case of $+$ linear polariation, $A_{+} = A = r(u)$ and $A_{\times} =0 $ while in the $\times$ linear polarization $A_{+} = 0,$ and $A_{\times} = ir(u)$. If the phase of $\Psi_4$ is constant in the complex plane this is called a linearly polarized wave \cite{Bicak, CV, ExactSolns}. 
\begin{subsection}{Plane wave spacetimes with $\bar{\gamma}_c = \gamma_c$}
\begin{lem} \label{lem:LinPolarization}
Relative to the null tetrad \eqref{InterpTet0}, if the phase $P$ of $\Psi_4$, defined 
as \beq P=\frac12 arctan\left( \frac{A_{\times}}{A_{+}} \right) 
\in (-\frac{\pi}{4},\frac{\pi}{4}), \label{Pfunk} \eeq \noindent  
is constant or $A_{+}=0$ (implying  $P = \pm \frac{\pi}{4}$), 
then the vacuum plane wave spacetime is linearly polarized with 
constant phase $2P$ and $\gamma_c$ must be real valued. In particular, 
if $P = 0$ the wave is in a $+$ linear polarization, and if 
$P = \pm \frac{\pi}{4}$ the wave is in a $\times$ linear polarization; each of these linear polarization modes are equivalent to each other via a spatial rotation. 
\end{lem}
\begin{proof}
To start we take the interpretation null tetrad \eqref{InterpTet0} and apply a spin and boost to produce the canonical coframe in which $\Psi_4 =1$. For any point along the arbitrarily chosen timelike geodesic, we may express the Cartan invariant $\gamma_c$ (relative to the class of canonical frames) in terms of $\tau$. Then, by expressing $\Psi_4$ relative to the original interpretation coframe \eqref{InterpTet0}, the non-vanishing Weyl tensor component \eqref{P4Interp} along the timelike geodesic becomes:
\beq \Psi_4 =  \dot{u} A(u(\tau)). \nonumber \eeq
\noindent Imposing the condition that $P=\frac12 arctan\left( \frac{A_{\times}}{A_{+}} \right)$  is constant so that $\Psi_4 = S(\tau)e^{i P}$, $S$ is a real valued function and
\beq \frac{\Psi^i_4}{\bar{\Psi}^i_4} =  \frac{A}{\bar{A}}. \nonumber \eeq 
This implies that $A$ must have constant phase $P$ in the complex plane;
by direct substitution into \eqref{PWBrinkmanGamma} we can show that
$\gamma_c$ is real-valued, so that this is indeed a linearly polarized
plane wave.  Rotating the coordinates $(\z, \bz)$ by $\theta/2$
\footnote{Equivalently applying a spin to the frame vectors $m$ and
$\bar{m}$.}  we may set $P =0$, and the plane wave is now + linearly
polarized.  Another rotation by $\pi/4$ will give the $\times$ linear
polarization, $P = \pi/2$.
\end{proof}
Thus, the two defining physical properties for a linearly polarized wave are the 
unchanging phase of the wave as $u(\tau)$ varies and the fact that the magnitude an 
observer measures depends on two functions: the value of $\dot{u}^2$ and the function  \beq \frac{|\Psi_4|}{\dot{u}^2} = A_{+}(u(\tau)) \nonumber \eeq 
\noindent as $u(\tau)$ varies along the worldline of the observer.
\end{subsection}
\begin{subsection}{Plane wave spacetimes with $\bar{\gamma}_c = - \gamma_c$}

 In section \ref{GDsec} we saw that the class of plane waves with the invariant, $\gamma_c$, a real valued scalar, corresponds to those plane waves in which any timelike geodesic gives rise to a linear polarization mode in the form of the geodesic deviation equations. We now consider the plane waves with the classifying function $\bar{\gamma}_c = -\gamma_c$, by expressing the metric in Brinkmann coordinates using Lemma \eqref{lem:G5App}. With a particular metric form we then examine the geodesic deviation equations relative to the complex null tetrad $\{ \ell, n, m, \bar{m} \}$ in the form of \eqref{GDnull}.

\begin{prop}\label{prop:Negamma}
Given a plane wave spacetime with $\bar{\gamma}_c = -\gamma$, then the form of $\gamma_c = ig(u)$ changes the polarization  in the following manner 
 \beq P(\tau) = \int -i\gamma_c(\tau) \dot{u} d\tau \nonumber \eeq
\end{prop}
\begin{proof}
Assuming $\gamma_c = i g(u)$, where $g$ is real-valued, we may without loss of generality use Lemma \eqref{lem:G5App} to integrate $f(\z,u)=A(u) \z^2$ and transform it to the following form:
\beq A(u) = e^{4\int i g(u) du} = e^{4 \int \gamma_c du}. \label{ImWavesMetric} \eeq
Applying a frame transformation of the form \eqref{ntot} to take the metric coframe to the interpretation coframe, the equations of geodesic deviation  relative to the complex null tetrad are 
\beq \ddot{z}^2 = \frac{\dot{u}^2}{2}e^{4\int i g(u) du} z^3. \label{ImWaveGDnull} \eeq
The sole components of the geodesic curve $u, \dot{u}$ are involved in the above equations, from which we see that the magnitude of the wave is directly related to $\dot{u}(\tau)^2$ as $\frac{|\Psi_4|}{\dot{u}^2} = 1.$  How the wave polarization varies is directly related to $P$, defined by the equation \eqref{Pfunk}, now a function of $\tau$ along the timelike geodesic:
\beq P(\tau) = \frac12 arctan(tan( 4 \int g du)) = 2 \int g(u(\tau)) \dot{u} d\tau \nonumber \eeq
\end{proof}
\end{subsection}
\end{section}
\begin{section}{The vacuum Plane Waves and the Rosen Form}
As another application of the classification we will use the transformation given in equation $(24.49)$ in \cite{ExactSolns} to switch from Brinkmann form to Rosen form \cite{Rosen, BPR, CV,CV0}. In light of the results in \cite{CV,CV0}, where a general formalism was introduced for studying arbitrary polarization states of PP-wave spacetimes with $\alpha = 0$ in Rosen coordinates, we would like to apply Lemma \eqref{lem:G5App} so that any novel solution found by this formalism may be expressed in Brinkmann coordinates.  

In Rosen coordinates the metric is written in the simple form:
\beq ds^2 = -2dudr + g_{AB}(u) dx^A dx^B,~~A,B \in [1,2]. \nonumber \eeq
\noindent The three functions involved in the symmetric $g_{AB}$ are connected by a differential constraint given by the vanishing of the only non-zero Ricci tensor component: 
\beq R_{AB} = - \left( \frac14 g''_{AB} - \frac14 g'_{AC} g^{CD} g'_{DB} \right), \nonumber \eeq
\noindent where differentiation with respect to $u$ is denoted by a prime. For a vacuum 
plane wave with arbitrary polarization, from the form of the metric in Brinkmann 
coordinates we require only two arbitrary functions of $u$, and so we expect this 
to hold true in the Rosen form of the metric as well.  

A particular anzatz for the metric functions in the Rosen metric was introduced by Bondi, Pirani and Robinson \cite{Bondi, BPR} to study gravitational plane waves:
\beq  ds^2 &=&  -e^{2Y}dudr +  u^2 cosh2Z(dx^2 + dy^2) \label{BPRmetric} \\
&&+  u^2 sinh2Z cos2W ( dx^2 - dy^2) - u^2sinh2Z sin2W(dxdy)]. \nonumber  \eeq 
\noindent where $Y(u),Z(u),W(u)$ satisfy 
\beq 2Y' = u(Z^{'2}+W^{'2}sinh2W). \nonumber \eeq
\noindent By examining the two independent components of the Riemann tensor, $\sigma$ and $\omega$ as defined in \cite{BPR}, the fixed plane polarization mode occurs if and only if $W =0$. In this case the metric simplifies to be
\beq ds^2 = -e^{2Y}dudr + u^2 [e^{2Z} dx^2 + e^{-2Z}]dy^2. \nonumber \eeq 
\noindent Choosing a new null coordinate $\tilde{u} = \int e^{2Y(u)} du$, this becomes the usual Rosen metric with $+$ linear polarization,
\beq ds^2 = -d\tilde{u}dr + \tilde{Y}(\tilde{u})^2 [e^{2Z} dx^2 + e^{-2Z}dy^2], \nonumber \eeq
\noindent where $\tilde{Y}$ denotes the inverse function of $e^{2Y}$. One may apply a rotation of the $(\z, \bz)$ coordinates to produce a $\times$ linear polarization or any other linear polarization mode of fixed phase. 

With that observation we have proven a helpful lemma.
%To deal with polarization modes, in particular the linear case, a canonical form for these metrics was introduced in  \cite{ExactSpacetimes}:
%\beq ds^2 = -2dudv + 2Y(u)^2 [ Z(u) dy^2 + Z(u)^{-1}(dx+W(u) dy)^2] \nonumber \eeq 
\begin{lem}  \label{lem:RosenLinPolarized}
In Rosen form, a vacuum plane wave is linearly polarized if and only if coordinates exist in which $W' = 0$.
\end{lem}
\noindent From which the results of the previous section imply:
\begin{cor}\label{lem:RosenLinPolarizedGamma} 
Relative to the canonical coframe, if $\gamma_c$ is real-valued, coordinates may be found in which $W' = 0$.
\end{cor}
Even in the simpler form \eqref{BPRmetric}, the plane waves in Rosen form are much more complicated than their Brinkmann counterparts. For example, in the case of linear polarization modes, the equations connecting $\tilde{Y}$ and $Z$ require considerably more analysis. This problem was studied in \cite{CV} and \cite{CV0}. Using the metric form \eqref{BPRmetric} along with the coordinate transformation $u' = \int e^{2Y}du$, the metric is now   
  \beq & ds^2 = -2 du dr + S^2(u) [ A'(u) dx^2 + 2B'(u) dx dy + C'(u) dy^2],& \nonumber \\
       & A' = cosh[X'(u)]+ cos[\theta'(u)] sinh[X'(u)],~~~B' = sin[\theta'(u)] sinh[X'(u)], & \label{WeberMetric} \\
       & C' = cosh[X'(u)]- cos[\theta'(u)] sinh[X'(u)].  \nonumber   & \nonumber \eeq
\noindent If $\theta = 0$ this metric describes linearly $+$ polarized waves, while if it is constant one has a linear polarization, along the axes produced by rotating by a fixed angle $\theta_0$. For example, setting $\theta = \frac{\pi}{2}$ yields the linearly $\times$ polarized waves. 

Since $A>0$ for all values of $u$, we may construct a null tetrad for this metric:
\beq \ell = du,~~n = dr,~~m = S(u) \left[ \sqrt{ C-\frac{B^2}{A} } dy + i \sqrt{A} \left(dx + \frac{B}{A} dy \right) \right] \label{RosenTet} \eeq
\noindent Relative to the class of coframes where $\Psi_4 = 1$, the spin-coefficient $\gamma_c$ is the only functionally independent invariant. The metrics describing $+$ and $\times$ polarizations produce a real value and purely imaginary value for $\gamma_c$, respectively, while an arbitrarily polarized wave will have $\bar{\gamma}_c \neq \pm \gamma_c$.
\end{section}
\begin{section}{An example: The weak-field circularly polarized waves}
 The circularly polarized waves were originally introduced as a weak-field solution, using the metric anzatz \eqref{WeberMetric} and requiring that $X_{,u} = 0$, $\theta_{,u} \neq 0$ and $\theta_{,uu} = 0$. These metrics were generalized to a class of strong field solutions \cite{CV,CV0} by requiring that $X = X_0$ and $\theta = \theta_0 u$ the metric becomes\footnote{ Notice that if $X_0 = 0$ or $\theta_0=0$ then $\Psi_4 = \Phi_{22} = 0$, the metric reduces to the Minkowski space. This is not possible since we have assumed $\Psi_4 \neq 0$, so these both must be non-zero.}
  \beq & ds^2 = -2 du dv + S^2(u) [ A(u) dx^2 + 2B(u) dx dy + C(u) dy^2],& \nonumber \\
& A = cosh[X_0]+ cos[\theta_0 u] sinh[X_0],~~~B = sin[\theta_0 u] sinh[X_0], & \label{CVMetric} \\
 & C = cosh[X_0]- cos[\theta_0 u] sinh[X_0].  & \nonumber \eeq
\noindent In these coordinates the sole non-vanishing components of the Ricci tensor is $R_{00}$; imposing vacuum conditions we find a form for $S$
\beq S = S_0 cos\left( \frac{sinh(X_0)\theta_0(u-u_0)}{2} \right). \nonumber \eeq

 In the strong field regime the construction of the Cartan invariants is considerably more involved. To provide a simple application of our work, we examine the weak field conditions by imposing $ X_0 <<1$ , we ignore all higher order terms and denote these as $O(X_0^2)$.

Thus for an arbitrarily long interval of $u$ the function $S$ may be approximated to be a constant $ S \approx S_0$, and without loss of generality we may always set $S_0 = 1$. Defining the following combinations of the functions $A,B,C$ in \eqref{CVMetric}, we may construct a null tetrad for the metric:
\beq &\ell = du,~n=dv,~~m = D^{\frac12} dy + i A^{\frac12}(dx+Edy).& \nonumber \\
& D = C-\frac{B^2}{A} = - \frac{-1+X_0^2}{1+cos(\theta_0u ) X_0} = \frac{1}{{1+cos(\theta_0u ) X_0}} &\nonumber \\
& E = \frac{B}{A} = \frac{sin(\theta_0 u)X_0}{1+cos(\theta_0u ) X_0}, & \nonumber \eeq
\noindent To see that this is approximately a Type N vacuum spacetime we calculate the sole component of the Curvature spinors which does not automatically vanish 
\beq & \Phi_{22} = -\frac14 \frac{X_0^2 \theta_0^2 (2cos(\theta_0 u) X_0 + X_0^2 + 1)}{1+2cos(\theta_0 u) X_0 + cos(\theta_0 u)^2 X_0^2}+O(X_0^2),& \label{RosenPs4} \\ 
&\Psi_4 = \frac{(i sin(\theta_0 u) - cos( \theta_0 u)) \theta_0^2 X_0}{2(1+cos(\theta_0 u) X_0)^2}+O(X_0^2)& \nonumber \eeq
\noindent imposing the weak-field condition it is clear that $\Phi_{22}$ does indeed vanish as $X_0^2 = 0$. 

To produce the Cartan invariants for these spacetimes we normalize $\Psi_4 =1$ by applying 
a boost using $a = \frac14 ln \Psi_4$. The transformation laws for spin coefficients produce $\gamma_c$ relative to this frame:
\beq &\gamma_c = \frac{1}{2 \sqrt{2} \sqrt{X_0}} \left( i X_0 cos(\theta_0 u) + \frac{2 X_0 sin(\theta_0 u)}{(1+cos(\theta_0 u) X_0)^3} + (1+cos(\theta_0 u) X_0) \right)+ O(X_0^2) & \label{Cgamma} \eeq
\noindent At second order, the invariant needed to fully classify the space is 
\beq& \Delta \gamma_c = \frac14 \left( i X_0 sin(\theta_0 u) + \frac{2X_0 sin(\theta_0 u)}{(1+cos(\theta_0 u)X_0)^3} + sin(\theta_0 u) X_0 \right) + O(X_0^2).&\label{DeltaCgamma}\eeq

As the constants $X_0$ and $\theta_0$ must both be non-zero, the combination  $Y = \sqrt{\frac{2}{X_0}} (\gamma_c - \bar{\gamma}_c)$ is a real-valued invariant with the simple form \beq Y = cos(\theta_0 u). \label{imgamma} \eeq   
\noindent We may locally express $\gamma_c$ and the second order invariant 
$\Delta \gamma_c$ in terms of $Y$, using some trigonometry and algebra. The original Cartan invariants at the first and second iteration of the Karlhede algorithm  produce two more classifying functions:
\beq & Re(\gamma_c) =  \frac{1}{2 \sqrt{2} \sqrt{X_0}} \left( \frac{2 X_0 \sqrt{1-Y^2}}{(1+Y X_0)^3} + (1+Y X_0) \right),~~\Delta Y =  \sqrt{1-Y^2} \nonumber \eeq
\noindent Of the pair of constants $(X_0,\theta_0)$, only $X_0$ uniquely determines the circularly polarized waves in the weak-field approximation, as all of the original Cartan invariants are written in terms of $Y$, 1 and $X_0$, while $\theta_0$ does not appear anywhere in the classifying functions. This can be confirmed by noting the coordinate transformation $u'=\theta_0u,~v'= \theta_0^{-1}v,~x'=x,~y'=y$ will produce a metric with $\theta_0 =1$ in general. 

For any plane-wave spacetime, along an arbitrary timelike geodesic the coordinates may be expressed in terms of the proper time, $\tau$. Then by applying a boost, spin and null rotation about $\ell$
we produce the coframe which is parallel transported along the curve \cite{Podolsky}.  Taking the spatial plane and using the null tetrad \eqref{InterpTet} with $\bar{z}^2 = z^3$, equation \eqref{GDnull} becomes:

\beq \ddot{z}^2 = \frac{-i \dot{u}^2 X_0 \theta_0^2 e^{-i \theta_0 u}}{4(1+cos(\theta_0 u)X_0)^2} z^3. \label{RosenGDnull} \eeq 
\noindent Here $\dot{u}$ and $u$ are the only functions of $\tau$ identifying a particular geodesic in the equations. The phase of the gravitational wave, $e^{-iC_0u}$, varies in a circular manner for any timelike geodesic, with $\theta_0$ dictating how quickly the phase spins as $u(\tau)$ changes due to the choice of coordinate system.  

Although the magnitude of the wave depends on the the timelike geodesic chosen 
\cite{Podolsky}, if the value of $\dot{u}^2$ is taken into account, the observer would notice the magnitude of the circular gravitational wave measured along the curve will vary as $u(\tau)$:
\beq \frac{|\Psi_4|}{\dot{u}^2} = \frac{\theta^2_0 X_0}{4(1+cos(\theta_0 u) X_0)^2}. \nonumber \eeq   
%Alternatively using an orthonormal basis for the spatial plane, if the Cartan invariant $Y$ is expressed in terms of $\tau$, $Y(u) = Y'(\tau) \in [-1,1)$, we may write the geodesic equations as
%\beq \ddot{Z}^2 &=& \frac{-\theta_0^2 X_0 \dot{u}^2 \sqrt{1-Y^{'2}}}{4(1+Y'X_0)^2} Z^2 - \frac{\theta_0^2 X_0 \dot{u}^2 Y'}{4(1+Y'X_0)^2} Z^3 \label{RosenGDEqnA}  \\
%\ddot{Z}^3 &=& \frac{-\theta_0^2 X_0 \dot{u}^2 Y'}{4(1+Y'X_0)^2} Z^2 + \frac{\theta_0^2 X_0 \dot{u}^2 \sqrt{1-Y^{'2}}}{4(1+Y'X_0)^2} Z^3. \label{RosenGDEqnB} \eeq
%\noindent It is unclear how this interpretation would be beneficial.
\end{section}

 \begin{section}{Conclusions}

We have used the Karlhede algorithm applied to the gravitational plane
wave spacetimes to produce a list of invariants arising from the
components of the curvature and its covariant derivatives.  We
identify the simplest set of invariants required to build the rest of
the invariants, and we can use this set to impose conditions on the
plane-waves and determine the essential invariants relating to
measurements a timelike observer might make in these spacetimes.

To show this we utilize the formalism introduced in \cite{Podolsky} to study the geodesic
deviation equations.  For a particular spacetime, these
equations describe how neighbouring timelike geodesics behave as one travels along the timelike geodesic; this quantities represent measurements a  timelike observer could make in the spacetime. In the case of the vacuum plane wave spacetimes, imposing conditions on the Cartan invariants lead to particular implications for the magnitutde and polarization of the plane waves, as measured by any timelike observer in the space.

In the future we hope to apply this to more general spacetimes in four
dimensions.  The vacuum Kundt waves \cite{McNutt} and all vacuum
plane-fronted gravitational waves in spacetimes with cosmological constant
\cite{Ozvath, Bicak} are natural candidates for future analysis, due to
the simple form the Riemann tensor takes in these spacetimes.
Alternatively, we could examine higher dimensional analogues of the
PP-waves and implement the generalization of the geodesic deviation
equations formalism to higher dimensions \cite{svarc} \end{section}

\begin{section}*{Acknowledgments} 
The authors would like to thank Jiri Podolsky for helpful comments. This work was supported by NSERC of Canada. 
\end{section}

\end{document}